\documentclass{article}

\usepackage[pdftex]{graphicx}
\usepackage{amssymb,amsmath}
\usepackage{graphicx}
\usepackage{tabularx}
\usepackage{epstopdf}
\usepackage{algorithmic}
\usepackage{algorithm}
\usepackage{setspace}

\newtheorem{theorem}{Theorem}

\newtheorem{definition}{Definition}
\newtheorem{lemma}{Lemma}
\newenvironment{proof}{\textbf{Proof }}{\begin{flushright} $\square$ \end{flushright}}

\title{How to Compute Times of Random Walks based Distributed Algorithms}
\author{Alain BUI \texttt{alain.bui@univ-reims.fr} \and Devan SOHIER \texttt{devan.sohier@univ-reims.fr}\\CRESTIC-SYSCOM\\D\'epartement de Math\'ematiques et Informatique \\Universit\'e de Reims Champagne Ardenne,\\ BP1039 F-51687 Reims cedex, France,}
\date{}
\begin{document}

\DeclareGraphicsExtensions{.eps}

\maketitle

\begin{abstract}
Random walk based distributed algorithms make use of a token that circulates in the system according to a random walk scheme to achieve their goal. To study their efficiency and compare it to one of the deterministic solutions, one is led to compute certain quantities, namely the hitting times and the cover time. Until now, only bounds on these quantities were defined.

First, this paper presents two generalizations of the notions of hitting and cover times to weighted graphs. Indeed, the properties of random walks on symmetrically weighted graphs provide interesting results on random walk based distributed algorithms, such as local load balancing. Both of these generalization are proposed to precisely represent the behaviour of these algorithms, and to take into account what the weights represent.

Then, we propose an algorithm to compute the $n^2$ hitting times on a weighted graph of $n$ vertices, which we improve to obtain a $O(n^3)$ complexity. This complexity is the lowest up to now. This algorithm computes both of the generalizations that we propose for the hitting times on a weighted graph.

Finally, we provide the first algorithm to compute the cover time (in both senses) of a graph. We improve it to achieve a complexity of $O(n^32^n)$. The algorithms that we present are all robust to a topological change in a limited number of edges. This property allows us to use them on dynamic graphs.
\end{abstract}

\section{Introduction}

The constant evolution of networking 
makes  it possible today to use several computers at a time to carry out a given computation. A distributed system is defined as a set of interconnected computing devices  called \emph{sites} or \emph{nodes},  cooperating in order to achieve a computation.
  A distributed system is usually modeled by a finite undirected graph $G(V,E)$, where $V$ is the set of sites  and $E$ is the set of communication links
(be either physical or logical).

This paper focuses on \emph{random walk based distributed algorithms}. These algorithms are \emph{token-based} algorithms - the token circulation 
 mechanism 
 is a well-known paradigm to achieve a global task in distributed computing. These algorithms have been designed to remove the strong hypotheses on the topology required by a deterministic token circulation scheme.
The token message circulates in the system, and at each step, the site that owns the token sends it to one of its neighbors chosen 
 at random.

The low message complexity makes token-based algorithms interesting, in comparison with the  flooding algorithms, the main interest  is their low time complexity. 
On many particular topologies, a deterministic token circulation scheme can be designed to efficiently visit  all the sites in the network: on a ring, the token can turn clockwise; on a chain, it can turn back and forth; on a tree, a depth first search provides positive results; and on a complete graph, the token can visit the sites according any strict order.
However, these schemes suffer of a lack of adptability  because they are designed for one particular topology and cannot be easily adapted to fit other ones. On the other hand, random walk based distributed algorithms can function on any topology, they require only a local knowledge of the topology (except for the standard assumption that the network remains connected).
Random walks offer an interesting property to adapt to the insertion or deletion of sites or links  in the network without modifying any of the code (as long as that the network remains connected; otherwise, no communication is possible between the connected components and the only solution is to launch one algorithm in each component). With the increasing dynamicity of networks, this feature is becoming crucial: redesigning a new browsing scheme at each modification of the topology is impossible, and flooding-based
procedures lead to the congestion of many networks.

The token circulation paradigm has been widely studied in the deterministic case.
Original solutions using random walks have been designed to solve various problems related to distributed computing \emph{e.g}  \cite{IsJa90} for self-stabilizing mutual exclusion, \cite{BDDN01} for mobile agent in wireless networks, \cite{BeBF04a} for token circulation in a dynamic and faulty environment or as an alternative   to flooding in decentralized and unstructured peer-to-peer networks (especially to achieve low bandwidth consumption 
  by control messages) \cite{LCCL+02}.

The time complexity of random walk based algorithms, like the one of deterministic token based algorithms, is the number of ``steps'' the algorithm takes to achieve the network traversal. Considering only one walk at a time (which is the case we are dealing with),
it is also equal to  the message complexity.

Random walk-based
 distributed algorithms must be then analyzed through probabilistic tools. It can be shown that a random walk will visit all the sites in a graph in a finite time, but there is no hope to give hard bounds to the time it will take: the classical worst-case analysis cannot be applied in this case, so we are led to use an average-case analysis.

The cover time $C$:  the average time to visit all nodes in the system, and the hitting time $h_{ij}$,  the average time it takes to reach a node $j$ for the first time starting from a given node $i$, are the two first important values that arise in the analysis of random walk-based distributed algorithms.
For instance, the cover time is the average time required to build a spanning tree thanks to the random walk token circulation algorithm \cite{Aldo90} and the hitting time is the average time it takes to  enter a critical section  \cite{IsJa90}. 

\paragraph{Related Works}

The mathematical background can be found in \cite{Lova93}. 

Many bounds on hitting times and cover times are available. \cite{BrWi90} Êproves that: $$h_{ij}\leq\frac4{27}n^3-\frac19n^2+ O(n)$$
\cite{Feig95} and \cite{Feig95a} show that: $$\left(1+o(1)\right)n\ln n\leq C\leq \left(\frac{4}{27}+o(1)\right)n^3$$
In \cite{Feig96}, the authors provide a polylogarithmic approximation for the cover time with a polynomial complexity. But the approximation can lead to severe biases, because they obtain a result of $(n-1)^2\ln n$ in place of $(n-1)^2$ on the path.
\cite{KKLV00} establishes that: $$\frac12M\leq C\leq 10^5M(\ln\ln n)^2$$ with $M=\max\{\kappa_S\ln|S|/S\subset V\}$, $\kappa_S=\max\{\kappa_{ij}/i,j\in S\}$ where $\kappa_{ij} = h_{ij} + h_{ji}$ is the commute time. The $\kappa_S$ provides a polynomial approximation within a factor 2 of $M$ (a naive computation of $M$ would have an exponential complexity since $|\mathcal P(V)|=2^{|V|}$). Even if this is a theoretically good approximation ($O((\ln\ln n)^2)$ means a slow
divergence), the actual ratio is $8\times 10^5\times (\ln\ln n)^2$, and the factor $8\times 10^5$ is very high with respect to $(\ln\ln n)^2$ in most concrete applications, especially in distributed computing.

\paragraph{Contributions}
All of these results, but the last one, do not take into account the topology of the graph, except for its size. Now, the topology explains the difference between the behavior of the walk on the various topologies, the \emph{``good''} cover time on the complete graph, and the \emph{``bad''} one on the lollipop graph (\emph{wrt} their sizes). In order to design insertion schemes and topologies in which the use of random walks is efficient, we provide  algorithms to compute the exact values of hitting and cover times in a graph. 
First, we extend the notion of  these two relevant values into a weighted graph, which illustrates a more general representation of distributed systems. The weight represents the quantity that hierarchizes the neighbors, and makes the token visit a neighbor rather than one another. It can be used to represent the bandwidth of a link, which will locally balance the load on the links.  Weights also appear where studying the average time it takes to reach a set of sites. For example, in file-sharing protocols, resources are generally replicated on several sites. To obtain the average time to first hit a site in the set $\mathcal O$ of all owners of the resource, we  consider the graph built from $G$ by removing $\mathcal O$ and adding a single site $o$. The weight of the link between $o$ and a site $i$ is the sum of the weights of all links between $i$ and a site of $\mathcal O$. Then, the hitting time from $i$ to $o$ is the quantity we were searching for.

 Then we propose an original algorithm to efficiently compute the hitting times. This algorithm provides a tool to compare the efficiency of random walk based distributed algorithms on large distributed networks. 
Finally, we propose a method to exactly compute the cover time. As far as we know, this is the first solution ever designed to solve this problem. Our method provides information better than the previously known  bounds presented above. Indeed, the result not only  takes topological informations  into account but it is also fairly robust (insertion or deletion of a limited number of communications links do not alter the cover time meaningfully).

\paragraph{Outline of the paper}
The first part of this article illustrates previous results on hitting and cover times. The second part presents a new efficient method to compute the hitting times between all pair of nodes in a graph with one matrix inversion. It requires some preliminary work to generalize previous results on hitting times. The third part presents the first algorithm to exactly compute the cover time of a graph, based on the hitting times computation presented before. Then, we conclude by offering some new perspectives.
In the sequel, we recall some results demonstrated by Chandra, Raghavan \emph{et al.}, and Tetali. We then derive some more general results, that we exploit to find an efficient algorithm to compute the hitting times in a graph. Finally, we offer an exemple of the execution of this algorithm.

\section{Preliminaries}

Random walks have been the subject of a wide applied mathematics litterature. Random walks are Markov chains, \emph{i.e.} memoryless stochastic process: if $(X_n)_{n\in\mathbb N}$ is a Markov chain ($X_n$ is the site that owns the token at time $n$) , $\forall n\in \mathbb N$:
$$P[X_{n+1}=\alpha|X_n=\alpha_n, X_{n-1}=\alpha_{n-1}, \ldots, X_0=\alpha_0]=P[X_{n+1}=\alpha|X_n=\alpha_n]$$

In this paper, 
the distributed system
topology is represented by a dynamic, connected, undirected, positively  real-weighted graph.  We denoted by ${\cal N}_i$ the set of neighbors of node $i$ (the set of nodes to which $i$ is connected).

The weight function will be denoted by $\omega$. For each edge $(i,j)$ a numerical value $\omega(i,j)$ is defined. Here the weight assignment is symmetric {\em i.e.} $\omega(i,j)=\omega(j,i)$. Consider the site $i$, we denote,  $\omega(i)=\sum_{j\in {\cal N}_i}\omega(i,j)$ and $\omega(G)=\sum_{\{i,j\}\in E}\omega(i,j)$. The graph is defined to be  unweighted if no weight assignment is assumed.

A random walk is then a sequence of nodes of $G$ visited by a token that starts at a node $i$ and visits other vertices according to the following transition rule: if a token is at $i$ at time $t$ then, at time $t+1$, it will be at one of the neighbors of $i$ chosen at random among all of them proportionally to the weight of the link adjacent to $i$.

A random walk on a weighted graph is such that, being on $i$, it moves from $i$ to $j$  at the next step with the probability:
$$P[X_{n+1}=j|X_n=i]=\frac{\omega(i, j)}{\omega(i)}$$
If the graph is unweighted, $\omega(i, j)=1$ if $i$ and $j$ are neighbours, $0$ else. The probability that the walk will eventually hit a given vertex is 1: starting on any site, the token will eventually hit a given site, even if it takes a long time (actually, no hard bound exists on this time).

According to the interpretation given to the weights on the links of the graph, the characteristic values presented above can each be given two definitions.

\begin{definition}
We call hitting time $h_{ij}$ \emph{in the first sense} (respectively cover time\emph{ in the first sense}) the average number of edges visited by the random walk starting at site $i$ to visit for the first time a site $j$ (respectively to cover the graph starting at site $i$).
\end{definition}

\begin{definition}
We call hitting time $h_{ij}$ in \emph{the second sense} (respectively cover time \emph{in the second sense}) the average total weight of edges visited by the random walk starting at site $i$ to visit for the first time a site $j$ (respectively to cover the graph starting at site $i$).
\end{definition}

The commute time is denoted by $\kappa_{ij} = h_{ij}+h_{ji}$.

Most of the previous results deal with the definitions in the first sense. In the next section, we extend the results to the second sense definition.

\paragraph{Random walks and resistances}\label{RW->R}

A tight link exists  between random walks and electrical networks  \cite{DoSn00}. 
We build an electrical network from a graph $G$ by replacing each of its edges by a resistor.  The conductance value (\emph{i.e.} the inverse of the resistance)  is equal to the weight of the edge in the graph it replaces (see figure\ref{figqque}).
\begin{figure}
\begin{center}
\includegraphics[width=.4\linewidth]{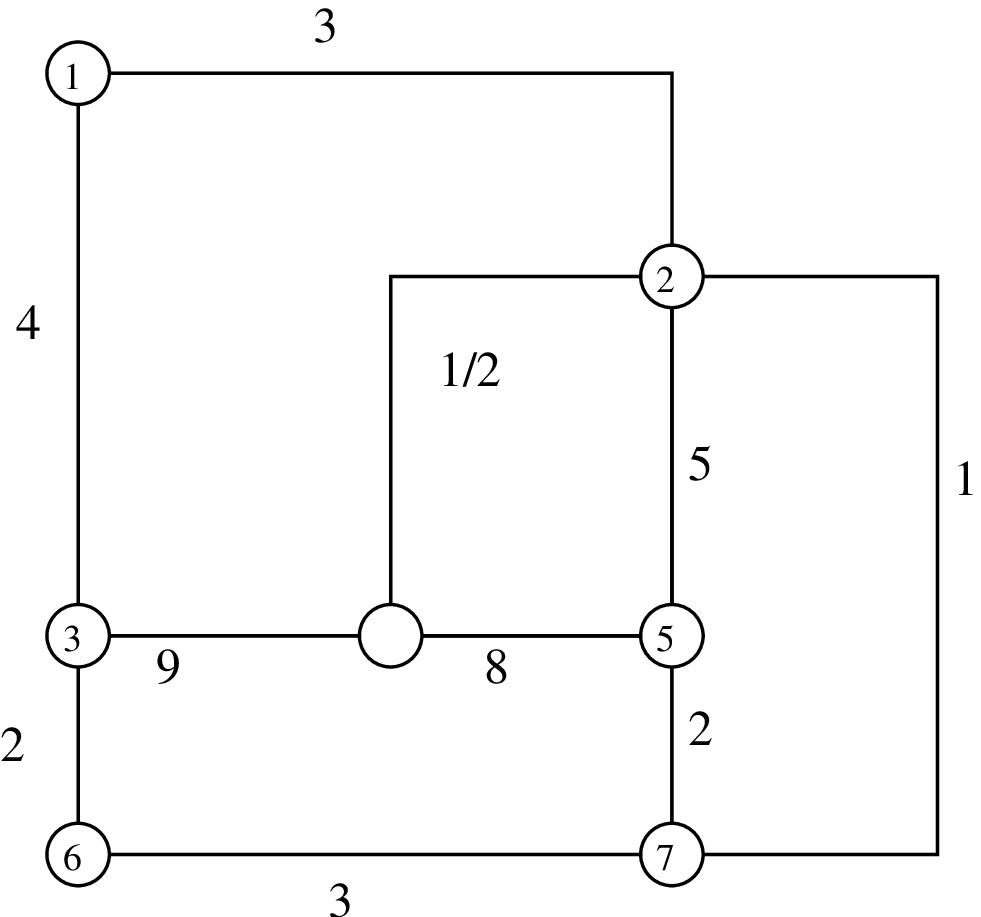}
\includegraphics[width=.4\linewidth]{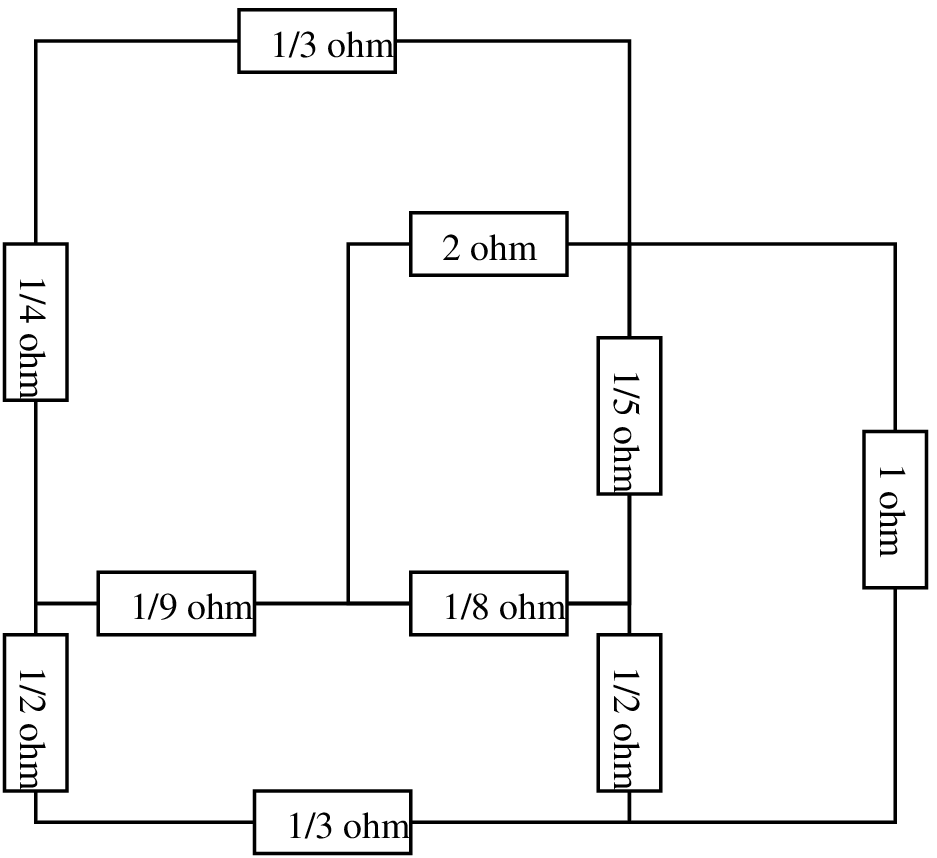}
\caption{the electrical circuit built from a weighted graph}
\label{figqque}
\end{center}
\end{figure}

Let $r({i,j})$ (resp. $c({i,j})$) denote the resistance (resp. conductance) of the resistor between two adjacent nodes $i$ and $j$. The equivalent resistance  $R_{ij}$  (resistance of the electrical network) between  $i$ and $j$ is defined as the resistance of the resistor to be placed between $i$ and $j$ to ensure the same electrical properties as the whole circuit.
 $R$ denotes the maximal equivalent resistance between two nodes of the network  \emph{i.e.}  $R= \max_{(i,j)\in V^2}R_{ij}$. 

\paragraph{Previous results for unweighted graph}

A tight relationship between resistances in electric networks and random walks characteristic  values as the hitting times and the cover time has been established \cite{CRRS+97}. In particular, it has been shown that, for random walks on unweighted graphs

\begin{lemma}
\begin{equation*}\label{equResistance} 
\kappa_{ij}=2mR_{ij}
\end{equation*}
\end{lemma}
  
where $i$ and $j$ denote two distinct vertices and $m$, the number of edges.

From this equation, we have: 

\begin{lemma}
\begin{equation*} mR<C<O(mR\log n) \end{equation*} 
\end{lemma}

In \cite{Teta91}, hitting times on unweighted graphs are expressed only in terms of resistances
\begin{lemma}\label{tetali}
\begin{equation*}
h_{ij}=mR_{ij}+\frac{1}{2}\sum_{k\in V}deg(k) \left(R_{jk}-R_{ik}\right)
\end{equation*}
\end{lemma}

\paragraph{Resistances computation}
Thanks to the Millman's theorem, we are able to compute all the resistances, as we have shown in \cite{BBBS03} 

\begin{theorem}[Millman's theorem]\label{millman}
Consider an electrical network, 
on any node $i$, the following relation holds: 

$$
V_i = \frac{\sum_{j\in {\cal N}_i}  \frac{V_j}{r_{ij}} }{\sum_{j\in {\cal N}_i}  \frac{1}{r(i,j)}}
$$

that is

$$\frac{V_i-V_{j_1}}{r(i,j_1)}+\frac{V_i-V_{j_2}}{r(i,j_2)}+\cdots+\frac{V_i-V_{j_n}}{r(i,j_n)} =0$$ 
where $j_1,\cdots,j_n$ are the neighbors of $i$, $V_{j_1},\cdots ,V_{j_n}$ are the
voltages of each of these nodes.
\end{theorem}

\section{Hitting times in the first and second sense}\label{hitting}

\subsection{Hitting times on weighted graphs}\label{conductance}

In \cite{BBBS03}, we provide an automatic way to compute resistance on unweighted graphs. Thanks to our method,  we can deduce from lemma \ref{tetali} above, the value of the hitting time between two nodes on such a graph. In this section, in order to generalize such a method , we establish  the relation between hitting times and equivalent resistances for weighted graph. 

\begin{theorem}
\begin{equation}\label{R->h}
h_{ij} = \omega(G)R_{ij} + \frac12\sum_{k\in G}\omega(k)(R_{jk}-R_{ik})
\end{equation}
\end{theorem}

\begin{proof}
The following reasoning is inspired by \cite{Teta91}.

Let $U_k^{ij}$ be the expected number of visits to $k$ in a random walk from $i$ to $j$. Then $U_j^{ij}=0$, and:\[U_k^{ij}=\sum_{l\in N(k)}U_l^{ij}p_{lk}=\sum_{l\in N(k)}U_l^{ij}\frac{\omega(l,k)}{\omega(l)}\]

Thus, \begin{equation}\label{U}\sum_{l\in N(k)}\omega(l,k)\frac{U_l^{ij}}{\omega(l)}=\omega(k)\frac{U_k^{ij}}{\omega(k)}\end{equation}

On the other hand, according to Kirchoff's current laws, when a unit current flows from $j$ to $i$, on each node $k$ except $i$ and $j$: ($V_k^{ij}$ denotes the potential on node $k$ when a unit current flows from $i$ to $j$ and $V_{kl}^{ij}=V_l^{ij}-V_k^{ij}$; for the sake of legibility, we may not write the superscript when it is obvious that the current flows from $i$ to $j$)
\begin{equation}\label{V}
\begin{split}
\sum_{l\in N(k)} c_{lk}V_{lk}&=0\\
\sum_{l\in N(k)} c_{lk}(V_{l}-V_{k})&=0\\
\sum_{l\in N(k)} c_{lk}V_{l}&=\left(\sum_{l\in N(k)} c_{lk}\right) V_{k}\\
\sum_{l\in N(k)} \omega(l,k)V_{l}&=\left(\sum_{l\in N(k)} \omega(l,k)\right) V_{k}\\
\sum_{l\in N(k)} \omega(l,k)V_{l}&=\omega(k) V_{k}
\end{split}
\end{equation}
From equations (\ref{U}) and (\ref{V}), since there is a single steady state in an electrical circuit,
we deduce that $\forall k, V_{kj}=\lambda\frac{U_k^{ij}}{\omega(k)}$, with $\lambda$ a factor such that the intensity circulating between $i$ and $j$ is 1,
\emph{i.e.} $\sum_{k\in N(j)}c_{kj}\lambda \frac{U_k^{ij}}{\omega(k)}=1$. If there were several solutions to (\ref{V}) with the same potentials on nodes $i$ and $j$, there would be several electrical steady states in this circuit.  Now, $\sum_{k\in N(j)}c_{kj}\lambda \frac{U_k^{ij}}{\omega(k)}=\lambda\sum_{k\in N(j)} U_k^{ij}\frac{\omega(k,j)}{\omega(k)}=\lambda$ for $\sum_{k\in N(j)} U_k^{ij}\frac{\omega(k,j)}{\omega(k)}$ is the average number of traversal toward $j$ of a random walk from $i$ to $j$, and this can only be 1. Thus, $\lambda=1$, and \begin{equation}\label{U=f(V)}V_{k}=\frac{U_k^{ij}}{\omega(k)}\end{equation}

Since $h_{ij}$ is the expected time for the walk to go from $i$ to $j$, $h_{ij}$ is the sum of the average number of visits of each site in the graph in the walk from $i$ to $j$. By linearity of the expectation, $h_{ij}=\sum_{l\in G}U_l^{ij}$. The expected \emph{commute time} $\kappa_{ij}$, which is the average time for a random walk to go from $i$ to $j$ and back is:
\begin{equation}\label{R->kappa}
\begin{split}
\kappa_{ij}&=\sum_{l\in G} U_l^{ij}+\sum_{l\in G} U_l^{ji}\\
&=\sum_{l\in G} \omega(l)(V_{lj}-V_{li})\\&=V_{ij}\sum_{l\in G}\omega(l)\\&=2\omega(G)R_{ij}
\end{split}
\end{equation}
where $R_{ij}$ is the equivalent resistance of the network between $i$ and $j$, \emph{i.e.} the voltage between those two nodes when a unit current enters $i$ and leaves $j$. Note that $U_l^{ji}=-V_{li}$, since the current goes from $i$ to $j$ instead of going from $j$ to $i$. 

Thanks to Kirchoff's current law applied on $j$, and using $V_i^{ik}=0$ (the potential is defined up to a constant, and we can assign an arbitrary value to one potential in the network; the resistance is defined by the difference between two potentials, this does not affect it), we have:
\[
V_j^{ik}=\frac{R_{jk}^{ijk}}{R_{jk}^{ijk}+R_{ij}^{ijk}}V_k^{ij}
\]
with $R_{ij}^{ijk}$, $R_{ik}^{ijk}$ and $R_{jk}^{ijk}$ the resistances of the resistors to be placed between $i$, $j$ and $k$ to ensure the same electrical properties as the original network: the resistances that give the potentials on any two of those three nodes allow us to obtain the third one like in the whole graph. As a result of applying this law on $k$, we obtain:
\[
V_k^{ik}=\frac{1+R_{jk}^{ijk}V_j^{ik}}{R_{kj}^{ijk}+R_{ik}^{ijk}}
\]
then:
\[
V_j^{ik}=\frac{R_{jk}^{ijk}}{R_{ij}^{ijk}R_{ik}^{ijk}+R_{ji}^{ijk}R_{jk}^{ijk}+R_{ki}^{ijk}R_{kj}^{ijk}}
\]

Thus, this formula being symmetrical in $j$ and $k$, $V_j^{ik}=V_k^{ij}$. So  \begin{equation*}
\begin{split}
R_{ik}&=V_{k}^{ik}\\&=V_{k}^{ik} + V_{j}^{ik} - V_{j}^{ik} \\&= V_{k}^{ik} + V_{k}^{ij} - V_{j}^{ik} \\&= V_{jk}^{ik} + V_{ik}^{ij}\\&=V_{jk}^{ik} + V_{ki}^{ji}\\&=\frac{U_j^{ik}}{\omega(j)}+\frac{U_k^{ji}}{\omega(k)}\end{split}\end{equation*}

Thus, \begin{equation*}
\begin{split}
R_{ik}+R_{kj}-R_{ij} &= \left(\frac{U_j^{ik}}{\omega(j)}+\frac{U_k^{ji}}{\omega(k)}\right)\\&\phantom{=} + \left(\frac{U_k^{ij}}{\omega(k)}+\frac{U_i^{jk}}{\omega(i)}\right)\\&\phantom{=} - \left(\frac{U_k^{ij}}{\omega(k)}+\frac{U_k^{ji}}{\omega(k)}\right) \\&= \frac{U_i^{jk}}{\omega(i)}+\frac{U_j^{ik}}{\omega(j)} \\&= 2\frac{U_k^{ij}}{\omega(k)}
\end{split}
\end{equation*}

Then,
\[
U_k^{ij}=\frac12 \omega(k)(R_{ij}+R_{jk}-R_{ik})
\]

and $h_{ij}=\sum_{k\in G}U_k^{ij}$
\end{proof}

\subsection{An efficient method to compute resistances}\label{solmat}

\paragraph{The basic method}
Our first solution (detailed in \cite{BBBS03}) to compute the equivalent resistance between two given nodes $i$ and $j$ consists in applying a 1V potential value on node $i$ and 0V on node $j$. $R_{ij}$  is obtained by the ratio of the potential difference between $V_i - V_j$ to the current circulating between these two nodes. The latter is established by the knowledge of the potentials at all the adjacent nodes of $i$ or $j$, given by the application of the Millman Theorem \ref{millman}.  An equivalent resistance is then computed by one matrix inversion (complexity $O(n^3)$) but $2n$ equivalent resistance computations are also necessary to obtain one hitting time by formula (\ref{R->h}).

We now propose an improved method. The basic idea comes from the observation that most of the matrices inverting in the previous method are similar. 

\paragraph{The improved method}
If we consider a 1A current injected in $i$ and flowing out in $j$, the Millman system can be rewritten as following:

\begin{equation*}
\left\{\begin{array}{l}
\forall k\in V\backslash\{i,j\}, \sum_{l\in N(k)} c_{kl} (V_k-V_l)=0\\
\sum_{l\in N(i)} c_{il} (V_i-V_l)=1\\
\sum_{l\in N(j)} c_{jl} (V_j-V_l)=-1
\end{array}\right.
\end{equation*}

This can be written:
\begin{equation*}
\Delta V=v
\end{equation*}
with $\Delta$ the matrix built from the conductance matrix by letting the entry $(k,k)$ be $-\sum_{l\in N(k)}c_{kl}$, and $v$ the vector with all entry 0 except the $i$-th one 1 and the $j$-th one -1.

However, $\Delta$ is not invertible. Indeed, in this system, no potential is specified and the potential is defined up to a constant. Thus we build a matrix $\Delta_2$ by replacing one of the lines of $\Delta$ (for instance the first one) by the corresponding line of the unity matrix, to set a potential in the system. 

Thus
\begin{equation*}
\Delta_2 V=v_2
\end{equation*}

where $v_2$  is obtained by replacing the first line 
(the same line index corresponding to the potential set) by an arbitrary value, 1 for instance.

Now, $\Delta_2$ is invertible, and $\Delta_2^{-1}v$ is a solution to $\Delta V=v$ that provides  the potentials on each node. 

The equivalent resistance between $i$ and $j$ is 
thus
$$
R_{ij} = \Delta_2^{-1}(i,j)-\Delta_2^{-1}(j,j)-\Delta_2^{-1}(i,i)+\Delta_2^{-1}(j,i).
$$

and consequently the hitting times can be deduced by formula  (\ref{R->h}).

The resistances between all pairs of nodes can thus be computed by inverting a single matrix, with a $O(n^3)$ complexity, while the basic method required  a $O(n^4)$ complexity.

The solution to compute the hitting times presented in \cite{KeSn76} is based on matrix computations and also has a $O(n^3)$ complexity. However, our method provides extra information that allow an efficient computation of the hitting times in the first sense and of the cover time, 
which the method in \cite{KeSn76} does not.

An increase in the conductance between two nodes in a circuit can only increase the global conductance of the circuit, and this by a factor of at most $1+\frac{r_{ij}}{R_{ij}}$. Thus, adding a new edge $(i, j)$ in graph can increase the commute time in a factor of at most $1+\frac{\omega(i, j)}{\omega(G)}$. This remark is known as Rayleigh's shortcut principle.

This principle shows that in a dense network, the adjonction or the removal of an edge does not modify much of the hitting time, namely, by at most $1+\frac{2\omega(i, j)}{\omega(G)}$. Thus, the computation of the hitting times of a graph provides results for a wide array of graphs deducing from the first one by removing or adding a few edges to it.

\subsection{Example}

\begin{figure}
\begin{center}
\includegraphics[width=.5\linewidth]{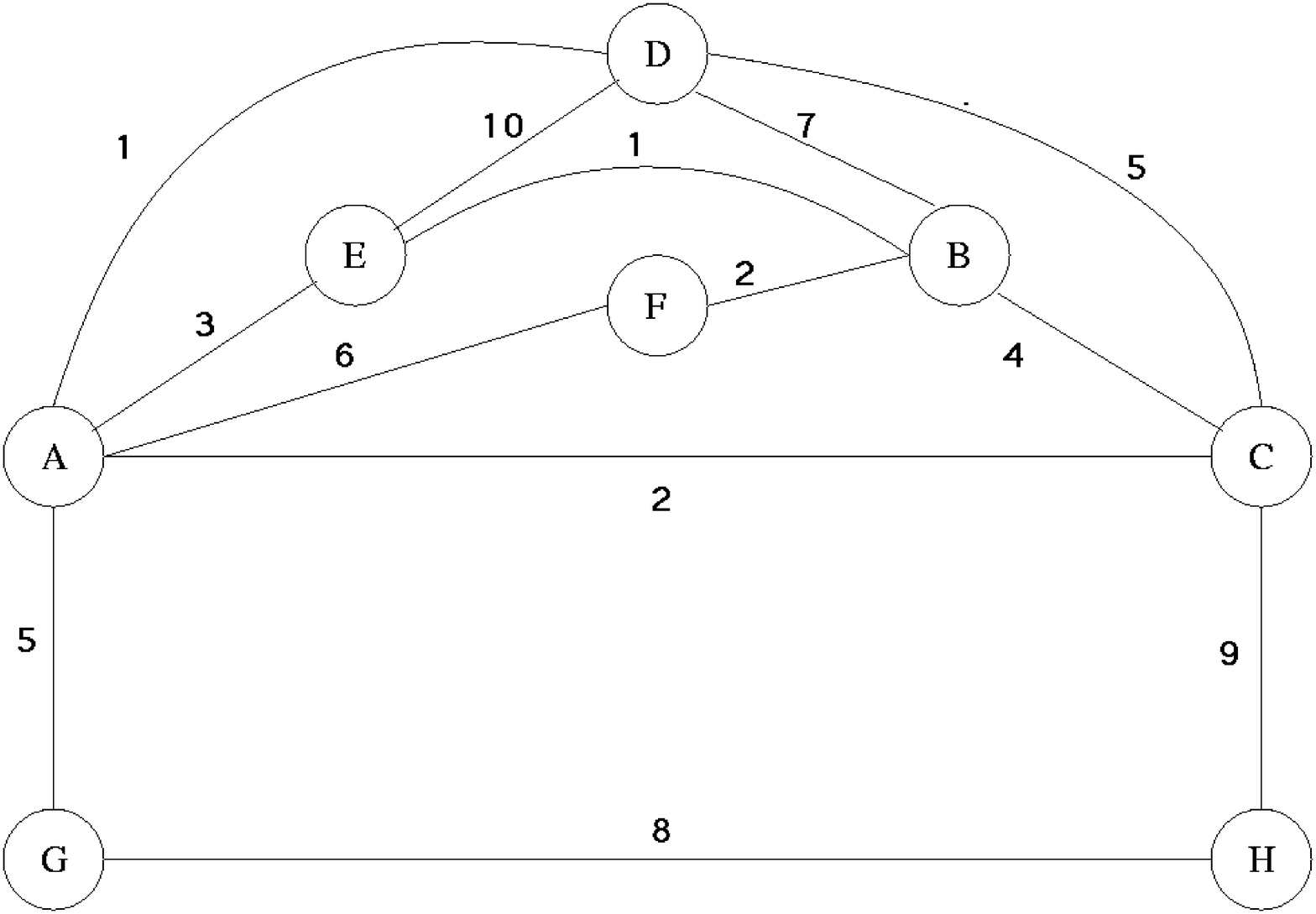}
\caption{example}
\label{ex}
\end{center}
\end{figure}

For the graph on figure \ref{ex}, the matrix $\Delta_2$ is:

\begin{center}
$\left (
\begin{tabular}{cccccccc}1 & 0 & 0 & 0 & 0 & 0 & 0 & 0 \cr 0 & -14 & 4 
& 7 & 1 & 2 & 0 & 0 \cr 2 & 4 & -20 & 5 & 0 & 0 & 0 & 9 \cr 1 & 7 & 5 &
     -23 & 10 & 0 & 0 & 0 \cr 3 & 1 & 0 & 10 & -14 & 0 & 0 & 0 \cr 6 & 2 
& 0 & 0 & 0 & -8 & 0 & 0 \cr 5 & 0 & 0 & 0 & 0 & 0 &
     -13 & 8 \cr 0 & 0 & 9 & 0 & 0 & 0 & 8 & -17 \cr
\end{tabular}
\right ) $
\end{center}

$\Delta_2^{-1}$ is:

 \setlength{\extrarowheight}{3 pt}
$$\left(\begin{array}{cccccccc}1 & 0 & 0 & 0 
\cr 1 & - \frac{102091}{627268}   & - \frac{56049}{627268}
       & - \frac{67339}{627268}   
\cr 1 & - \frac{56049}{627268}       & - \frac{84309}{627268}   & - \frac{53851}{627268}   
       \cr 1 & - \frac{67339}{627268}   & - \frac{53851}{627268}   & - 
\frac{89297}{627268}
\cr 1 & - \frac{110783}{1254536}   & - \frac{84937}{1254536}
        & - \frac{137187}{1254536}  
\cr 1 & - \frac
        {102091}{2509072}   & - \frac{56049}{2509072}   & - 
\frac{67339}{2509072} 
 \cr 1 & - \frac{6426}{156817}   & - 
\frac{9666}{156817}   & -
      \frac{6174}{156817}
\cr 1 & - \frac{41769}{627268}   & - \frac{62829}{627268}   & - 
\frac{40131}{627268}
\end{array}
\right.
\left.
\begin{array}{cccccccc} 0 & 0 & 0 & 0 \cr  - \frac{110783}{1254536}   & - 
\frac{102091}{2509072}
       & - \frac{6426}{156817}   & - \frac{41769}{627268}   \cr - 
\frac{84937}{1254536}
       & - \frac{56049}{2509072}   & - \frac{9666}{156817}   & - 
\frac{62829}{627268}
       \cr  - \frac{137187}{1254536}   & - \frac{67339}{2509072}   & - 
\frac{6174}{156817}
       & - \frac{40131}{627268}   \cr    - \frac{391027}{2509072}   & - 
\frac{110783}
      {5018144}   & - \frac{4869}{156817}   & - \frac{63297}{1254536}   
\cr  -
      \frac{110783}{5018144}   & - \frac{1356627}{10036288}   & - 
\frac{3213}{313634}   & -
      \frac{41769}{2509072}   \cr  - \frac{4869}{156817}   & - 
\frac{3213}{313634}   & -
       \frac{21413}{156817}   & - \frac{15194}{156817}   \cr  - \frac{63297}{1254536}   & -
      \frac{41769}{2509072}   & - \frac{15194}{156817}   & - 
\frac{98761}{627268}   \cr
\end{array}\right)$$

The resistance matrix is:

\setlength{\extrarowheight}{3 pt}
$$\left(\begin{array}{cccc}0 & \frac{102091}{627268} & 
\frac{84309}{627268} & \frac{89297}{627268}  \cr 
\frac{102091}{627268} & 0 & \frac{37151}{313634} & \frac{28355}
    {313634} \cr \frac{84309}
    {627268} & \frac{37151}{313634} & 0 & \frac{16476}{156817} 
\cr \frac{89297}{627268} & 
\frac{28355}{313634} & \frac{16476}{156817} & 0 
\cr \frac{391027}{2509072} & \frac{356259}
    {2509072} & \frac{388515}{2509072} & \frac{199467}{2509072} 
 \cr \frac{1356627}{10036288} & 
\frac{2173355}{10036288} & \frac{2257179}{10036288} & \frac{2246667}
    {10036288} 
  \cr \frac{21413}
    {156817} & \frac{136335}{627268} & \frac{92633}{627268} & 
\frac{125557}{627268}  \cr \frac{98761}{627268} & 
\frac{58657}{313634} & \frac{14353}{156817} & \frac{26949}
    {156817}
\end{array}\right.
\left.\begin{array}{cccc} \frac{391027}{2509072} & 
\frac{1356627}
    {10036288} & \frac{21413}{156817} & \frac{98761}{627268} \cr 
 \frac{356259}{2509072} & \frac{2173355}{10036288} & 
\frac{136335}{627268} & \frac{58657}{313634} \cr 
\frac{388515}{2509072} & \frac{2257179}{10036288} & \frac{92633}
    {627268} & \frac{14353}{156817} \cr  \frac{199467}
    {2509072} & \frac{2246667}{10036288} & \frac{125557}{627268} & 
\frac{26949}{156817} \cr  0 & 
\frac{2477603}{10036288} & \frac{577827}{2509072} &
     \frac{532883}{2509072} \cr  \frac{2477603}{10036288} & 0 & \frac{2521427}{10036288} 
& \frac{2602651}{10036288} \cr \frac{577827}{2509072} & \frac{2521427}
    {10036288} & 0 & \frac{62861}{627268} \cr  \frac{532883}{2509072} & \frac{2602651}{10036288} & 
\frac{62861}{627268} & 0 \cr
\end{array}\right)$$

and the hitting times matrix is:
  \setlength{\extrarowheight}{3 pt}
$$\left(\begin{array}{cccc}0 & \frac{6740527}{627268} & 
\frac{4636113}{627268} & \frac{5079563}{627268} 
\cr 
\frac{6122939}{627268} & 0 & \frac{1848439}{313634} &
     \frac{679447}{156817} 
\cr \frac{5986821}{627268} & \frac{2832587}{313634} & 0 & 
\frac{2140579}{313634} 
 \cr \frac{6171859}{627268} & \frac{1106918}
    {156817} & \frac{2011373}{313634} & 0 
  \cr \frac{10946183}{1254536} & 
\frac{10468579}{1254536} & \frac{9516347}{1254536} & \frac{3819747}
    {1254536} 
   \cr \frac{8632011}
    {2509072} & \frac{22730653}{2509072} & \frac{20114289}{2509072} & 
\frac{20465549}{2509072} 
\cr \frac{1042971}{156817} & \frac{7673707}{627268} &
     \frac{3936333}{627268} & \frac{6139751}{627268} 
\cr \frac{5760001}{627268} & \frac{3618817}{313634} & 
\frac{619915}{156817} & \frac{2891529}{313634} 
    \end{array}\right.
\left.\begin{array}{cccccccc} 
\frac{6844259}{627268} & \frac{68203479}
    {5018144} & \frac{1655067}{156817} & \frac{6683885}{627268} \cr 
 \frac{5987869}{627268} & 
\frac{91460059}{5018144} & \frac{9504503}{627268} & \frac{3771965}
    {313634} \cr  \frac{7480049}{627268} & \frac
      {101973699}{5018144} & \frac{7735425}{627268} & 
\frac{1188563}{156817} \cr \frac{4373337}{627268} & 
\frac{100608923}{5018144} & \frac{9680431}{627268} &
     \frac{3899619}{313634} \cr 0 & \frac{98029553}{5018144} & \frac{19278767}{1254536} 
& \frac{16338531}{1254536} \cr \frac{14514859}
    {1254536} & 0 & \frac{31874379}{2509072} & \frac{30104657}{2509072} 
\cr 
\frac{8562167}{627268} & \frac{95101143}{5018144} & 0 & \frac{3197993}
    {627268} \cr  \frac{8616549}
    {627268} & \frac{103757699}{5018144} & \frac{4722493}{627268} & 0 \cr
    \end{array}\right)$$

\subsection{Hitting times in the first sense}

Hitting times in the first sense can then be computed thanks to: ($n_k$ is the expected number of visit to $k$ in a walk from $i$ to $j$)
$$h_{ij}=\omega(i)R_{ij}\sum_{k\in V}R_{ki}$$

Indeed, the electrical potential is proportionnal to the average number of visits  $n_i$ (\cite{Lova93}) to a site $i$ in a walk from the site of potential 1 to the site of potential 0 (\cite{Boll98}), since $$n_i=\sum_{j\in\mathcal N(i)}n_j\times \frac{\omega(i)}{\omega(i, j)}$$ and the potential also solves this equation. All of these solutions being proportionnal, $n_i=n_1V_i$ if 1 is the site with potential 1. Now, the probability that the token hits 0 before going back to 1 is $\frac{2\omega(G)}{\omega(1)\kappa_{10}}$, as shown in \cite{Lova93}. The expected number of return to 1 before hitting 0 is thus:
\begin{equation*}
\begin{split}
n_1&=\sum_{k=1}^{+\infty}k\left(1-\frac{2\omega(G)}{\omega(1)\kappa_{10}}\right)^k\frac{2\omega(G)}{\omega(1)\kappa_{10}}\\
&=\sum_{k=1}^{+\infty}\sum_{l=1}^k\left(1-\frac{2\omega(G)}{\omega(1)\kappa_{10}}\right)^k\frac{2\omega(G)}{\omega(1)\kappa_{10}}\\
&=\sum_{l=1}^{+\infty}\sum_{k=l}^{+\infty}\left(1-\frac{2\omega(G)}{\omega(1)\kappa_{10}}\right)^k\frac{2\omega(G)}{\omega(1)\kappa_{10}}\\
&=\frac{2\omega(G)}{\omega(1)\kappa_{10}}\times\frac{1}{1-\left(1-\frac{2\omega(G)}{\omega(1)\kappa_{10}}\right)}\times\sum_{l=1}^{+\infty}\left(1-\frac{2\omega(G)}{\omega(1)\kappa_{10}}\right)^l\\
&=\frac{\omega(1)\kappa_{10}}{2\omega(G)}
\end{split}
\end{equation*}

Thus, $n_k=n_i\frac{V_k-V_i}{V_j-V_i}$ ($V_i=0$, $V_j=1$)
and,
\begin{equation*}
\begin{split}
h'_{ij}  & =\sum_{k\in V}n_k\\
&=n_i\sum_{k\in V}\frac{V_k-V_i}{V_j-V_i}\\
&=\frac{\omega(i)\kappa_{ij}}{2\omega(G)}\sum_{k\in V}\frac{V_k-V_i}{V_j-V_i}\\
&=\omega(i)(V_j-V_i)\sum_{k\in V}\frac{V_k-V_i}{V_j-V_i}\\
&=\omega(i)R_{ij}\sum_{k\in V}R_{ki}
\end{split}\end{equation*}

\subsection{Cyclic cover time}

The cyclic cover time is defined as:$$\min\left\{\left.\sum_{i=1}^{n}h_{i\sigma(i)}\right/\sigma\in\mathfrak S_n\right\}$$
with $\mathfrak S_n$ the cyclic group of order $n$.

The cyclic cover time is an upper bound of the cover time. It represents the average time for a walk to visit all vertices in the best deterministic order. \cite{Feig95, CoFS96} make use of the cyclic cover time to bound the cover time.

\cite{CoFS96} computes the cyclic cover time thanks to a travelling salesperson formulation, which includes a prior computation of all the hitting times. Our algorithm can time-quiclky speed up the first phase of this computation.

\section{Computation of the Cover Time}

The cover time is the expected time for a random walk starting from a given node to visit all the nodes in  the graph. In terms of random walk based distributed algorithm, this is the time required to broadcast a piece of information to all computers taking part in the process. In the algorithm in \cite{Aldo90} that  builds a spanning tree, the cover time is the average time after which the algorithm has built a spanning tree (note that some fault-tolerant algorithms are based on it, the stabilization time is the cover time of the graph).

In this  section, we first reformulate the problem in terms of hitting times on a graph $\mathcal G$, then give an algorithm providing the cover time. This algorithm is improved in the next subsection and we conclude by providing an example of this.

To compute the cover time, we need a criterion to determine whether every vertex has been visited by the token. Consider  $G=(V,E)$ the undirected connected graph modeling a distributed system. We build from $G$ an associated graph $\mathcal G$  so that the cover time of $G$ can be expressed in terms of hitting times in $\mathcal G$. To express results on cover time using hitting times, we have to take into account the token trajectory. So $\mathcal G$ should reflect some history-dependant data. 

In this section, we limit the reasoning to unweighted graph to avoid big equations meaninglessly. Nevertheless, all of them hold with weighted graphs.

\subsection{Construction of the associated graph $\mathcal{G}$}

First let define $\mathcal G =(\mathcal V, \mathcal E) $ where $\mathcal{V}$ is a set of nodes and $\mathcal{E}$ a set of directed edges.  

\begin{itemize}
\item $x \in \mathcal{V}$ is defined by $x = (P,i)$ with $P \in \mathcal{P}(V)$ where $\mathcal{P}(V)$ is the power set of V (set of nodes of $G$) and $i\in V$. $P$ represents the set of nodes in $G$ already visited by the token,  and $i\in V$ the vertex on which the token is currently on.  

\item any edge $(x,y) \in \mathcal{E} $ is of the form $(x,y) = ((P,i), (Q,j))$ with $(x,y) \in \mathcal{V} \times \mathcal{V}$ and $(i,j) \in E$.
\end{itemize}

Suppose that, initially,  the token is at node $i$ in $G$, and next the token moves to $j$ neighbor of $i$, and then next moves back to $i$. In the associated graph $\mathcal{G}$ , we have the following path
$
\left(  (\{  i \}, i) ; (\{  i ,j\}, j) ; (\{  i, j \}, i) \right )$.

Note that $\mathcal{E}$ is a set of  \emph{directed} edges $((P,i), (Q,j))$. Edges in $\mathcal{E}$ are defined by:

\begin{itemize}
\item $((P,i), (P,j))$, where $i\in P$ and $j\in P$ are neighbors; this case corresponds to a token transmission to the node $j$ which has already been visited by the token.

\item $((P,i), (P\bigcup\{j\},j))$ where $i\in P$ and $j\notin P$ are neighbors; this case corresponds to a token transmission to the node $j$ which is holding the token for the first time.
\end{itemize}

The probability to obtain a given path in $G$ is equal to the probability to obtain the associated path in  $\mathcal G$. Indeed, for $i \in P\subset V$ and $j\in V$, there exists some $Q \subset V$ such that the transition probability from $(P,i)$ to $(Q,j)$ and the transition probability from $i$ to $j$ are equal: $Q=P$ if $j\in P$, else, $Q=P\bigcup\{j\}$.

A token in $G$ has visited every node \emph{iff} the associated token in $\mathcal G$ has reached a node $(P,i)$ such that $P=V$. Then, we deduce that the cover time in $G$ is the average time it takes to a  token in $\mathcal{G}$ starting from a node $i$ to reach any arbitrary node $k$ for the first time while having visited all nodes, that is

$$C_i(G)=h_{(\{i\},i), \{(V,k)/k\in V\}}(\mathcal G)$$ 

The token has covered $G$ when the associated token in $\mathcal{G}$ has hit any vertex
 in $F=\{(V,k)/k\in V\}$. We do not care at which node $(V,k)$ the token reaches in $\mathcal{G}$, then we lump all nodes in $F$ into a single node called $f$ (in fact we obtain an absorbing Markov Chain). Now, the cover time in $G$ is obtained by the average  number of steps needed before entering $f$ starting in node $(\{ i\}, i)$.

\subsection{Cover time computation}

$\mathcal G$ being directed, we cannot apply the procedure in section \ref{hitting} to compute $h_{(\{i\},i), \{(V,k)/k\in V\}}(\mathcal G)$. 

Let $\mathcal N_o(x)$ be the set of vertices that have an incoming edge from $x$: $\{y\in \mathcal V/(x,y)\in \mathcal E\}$.

Since $f$ can be reached from any vertex (if not, some of the $h_{xf}$ would be undefined) we have, \begin{equation}\label{syslinh}\left\{\begin{array}{l}\forall x\in V, h_{xf}=1+\sum_{y\in\mathcal N_o(x)}p_{xy}h_{yf}\\ h_{ff}=0\end{array}\right.\end{equation}

The square linear system  (\ref{syslinh})  has a single solution (vector $h_{. f}$) then the hitting time between all nodes and a given node can be computed by inverting one matrix. 
 
Thus,  the cover time of any graph $G$ is computed by building $\mathcal G$ and by computing $h_{(i,\{i\}), f}(\mathcal G)$, which requires the inversion of an approximatively $n2^n\times n2^n$ matrix.

Let $G$ be the graph on figure \ref{gr4}. Then $\mathcal G$ is partially represented by the graph on figure \ref{grC}.

\begin{figure}
\begin{center}
\begin{minipage}[b]{.3\linewidth}
\includegraphics[width=.9\linewidth]{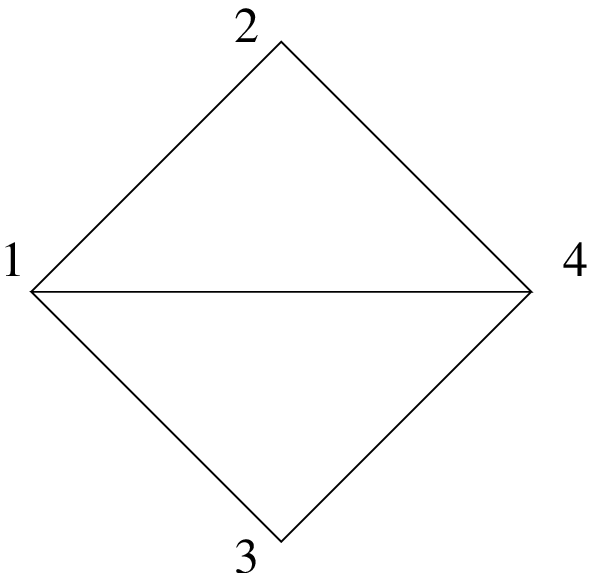}
\caption{$G$\label{gr4}}
\end{minipage}
\begin{minipage}[b]{.45\linewidth}
\includegraphics[angle=270, width=.9\linewidth]{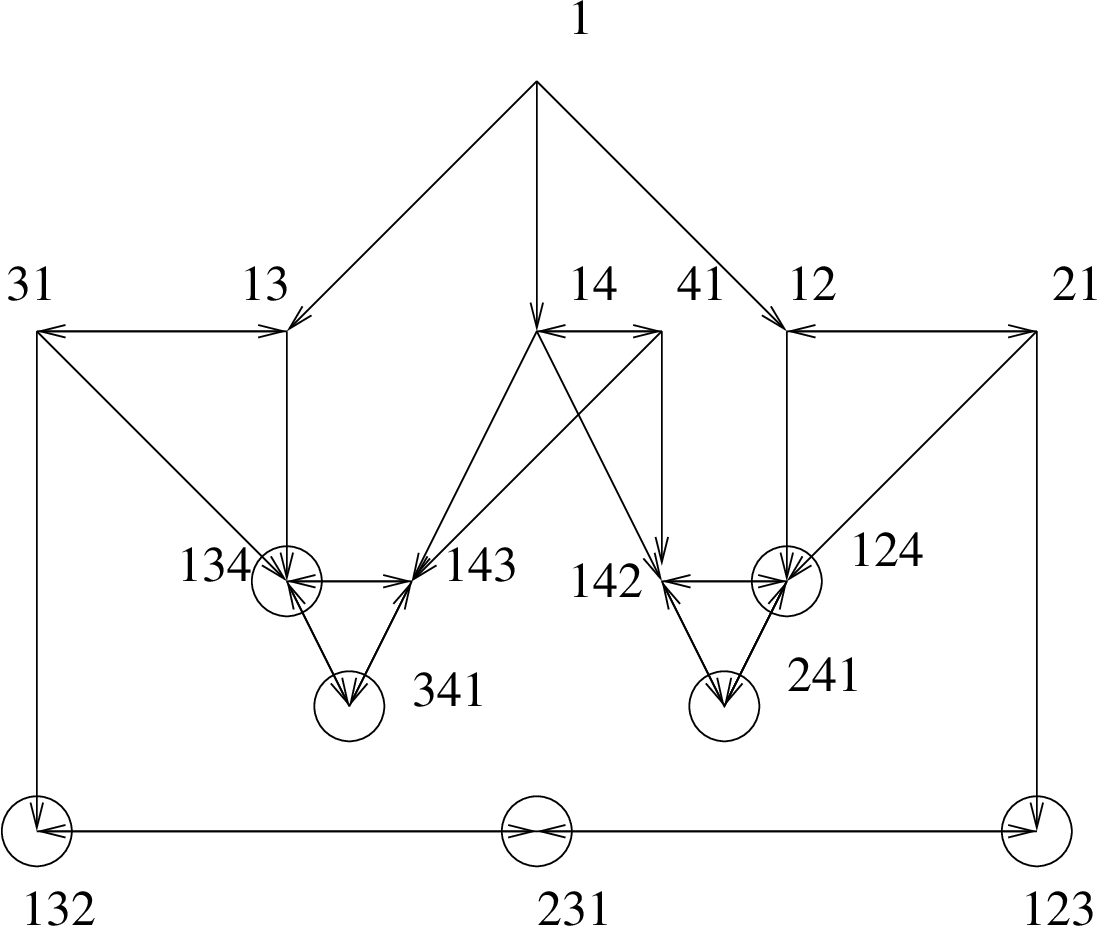}
\caption{$\mathcal G$\label{grC}}
\end{minipage}
\end{center}
\end{figure}

In figure \ref{grC}, we use the following notation: $ij\textbf{k}$ corresponds to node $(\{i,j,\mathbf{k}\}, \mathbf{k})$ (\emph{e.g.} 31 corresponds to $(\{ 1,3\}, 1)$ and 13 corresponds to $(\{ 1,3\}, 3)$).
 We only built the part of $\mathcal G$ that corresponds to situations where the token started in node 1.  We did not write the states in which all vertices are visited: for the sake of legibility, we circled the sites that lead to such a state. Thus, in state 134, the token will reach 2 and achieve to cover the graph with probability $\frac13$, reach 3 (the state being 143) or 1 (341) also with probability $\frac13$.

Since we merge all the states in which the token has covered the graph, every circled state leads to the new site $f$ with a directed vertex. We did not write unreachable sites.

The matrix of $\mathcal G$ is then:

{\small $$\left(\begin{array}{r|c c c c c c c c c c c c c c c c c}
&1 & 12 & 14 & 13 & 21 & 124 & 41 & 142 & 143 & 31 & 134 & 123 & 132 & 241 & 341 & 231 & f \cr
\hline
1&   0 & 1 & 1 & 1 & 0 & 0 & 0 & 0 & 0 & 0 & 0 & 0 & 0 & 0 & 0 & 0 & 0 \cr
12&   0 & 0 & 0 & 0 & 1 & 1 & 0 & 0 & 0 & 0 & 0 & 0 & 0 & 0 & 0 & 0 & 0 \cr
14&   0 & 0 & 0 & 0 & 0 & 0 & 1 & 1 & 1 & 0 & 0 & 0 & 0 & 0 & 0 & 0 & 0 \cr
13&   0 & 0 & 0 & 0 & 0 & 0 & 0 & 0 & 0 & 1 & 1 & 0 & 0 & 0 & 0 & 0 & 0 \cr 
21&   0 & 1 & 0 & 0 & 0 & 1 & 0 & 0 & 0 & 0 & 0 & 1 & 0 & 0 & 0 & 0 & 0 \cr 
124&   0 & 0 & 0 & 0 & 0 & 0 & 0 & 1 & 0 & 0 & 0 & 0 & 0 & 1 & 0 & 0 & 1 \cr 
41&   0 & 0 & 1 & 0 & 0 & 0 & 0 & 1 & 1 & 0 & 0 & 0 & 0 & 0 & 0 & 0 & 0 \cr 
142&   0 & 0 & 0 & 0 & 0 & 1 & 0 & 0 & 0 & 0 & 0 & 0 & 0 & 1 & 0 & 0 & 0 \cr 
143&   0 & 0 & 0 & 0 & 0 & 0 & 0 & 0 & 0 & 0 & 1 & 0 & 0 & 0 & 1 & 0 & 0 \cr 
31&   0 & 0 & 0 & 1 & 0 & 0 & 0 & 0 & 0 & 0 & 1 & 0 & 1 & 0 & 0 & 0 & 0 \cr 
134&   0 & 0 & 0 & 0 & 0 & 0 & 0 & 0 & 1 & 0 & 0 & 0 & 0 & 0 & 1 & 0 & 1 \cr 
123&   0 & 0 & 0 & 0 & 0 & 0 & 0 & 0 & 0 & 0 & 0 & 0 & 0 & 0 & 0 & 1 & 1 \cr 
132&   0 & 0 & 0 & 0 & 0 & 0 & 0 & 0 & 0 & 0 & 0 & 0 & 0 & 0 & 0 & 1 & 1 \cr 
241&   0 & 0 & 0 & 0 & 0 & 1 & 0 & 1 & 0 & 0 & 0 & 0 & 0 & 0 & 0 & 0 & 1 \cr 
341&   0 & 0 & 0 & 0 & 0 & 0 & 0 & 0 & 1 & 0 & 1 & 0 & 0 & 0 & 0 & 0 & 1 \cr 
231&   0 & 0 & 0 & 0 & 0 & 0 & 0 & 0 & 0 & 0 & 0 & 1 & 1 & 0 & 0 & 0 & 1 \cr 
f&   0 & 0 & 0 & 0 & 0 & 0 & 0 & 0 & 0 & 0 & 0 & 0 & 0 & 0 & 0 & 0 & 0 \end{array}\right)
$$

}
The system that we have to solve to obtain the cover time is:
{\small $$
  \setlength{\extrarowheight}{3 pt}\left(\begin{array}{c c c c c c c c c c c c c c c c c}
1 & - \frac{1}{3}   & - \frac{1}{3}   & - \frac{1}{3} 
      & 0 & 0 & 0 & 0 & 0 & 0 & 0 & 0 & 0 & 0 & 0 & 0 & 0 \cr 0 & 1 & 0 & 0 & - \frac{1}{2}   & - \frac{1}
     {2}   & 0 & 0 & 0 & 0 & 0 & 0 & 0 & 0 & 0 & 0 & 0 \cr 0 & 0 & 1 & 0 & 0 & 0 & - \frac{1}{3}   & -
     \frac{1}{3}   & - \frac{1}{3} 
      & 0 & 0 & 0 & 0 & 0 & 0 & 0 & 0 \cr 0 & 0 & 0 & 1 & 0 & 0 & 0 & 0 & 0 & - \frac{1}{2}   & - \frac{1}
     {2}   & 0 & 0 & 0 & 0 & 0 & 0 \cr 0 & - \frac{1}{3}   & 0 & 0 & 1 & - \frac{1}{3} 
      & 0 & 0 & 0 & 0 & 0 & - \frac{1}{3}   & 0 & 0 & 0 & 0 & 0 \cr 0 & 0 & 0 & 0 & 0 & 1 & 0 & - \frac{1}
     {3}   & 0 & 0 & 0 & 0 & 0 & - \frac{1}{3}   & 0 & 0 & - \frac{1}{3}   \cr 0 & 0 & -
     \frac{1}{3}   & 0 & 0 & 0 & 1 & - \frac{1}{3}   & - \frac{1}{3} 
      & 0 & 0 & 0 & 0 & 0 & 0 & 0 & 0 \cr 0 & 0 & 0 & 0 & 0 & - \frac{1}{2}   & 0 & 1 & 0 & 0 & 0 & 0 & 0 & -
     \frac{1}{2}   & 0 & 0 & 0 \cr 0 & 0 & 0 & 0 & 0 & 0 & 0 & 0 & 1 & 0 & - \frac{1}{2}   & 0 & 0 & 0 & 
    - \frac{1}{2}   & 0 & 0 \cr 0 & 0 & 0 & - \frac{1}{3}   & 0 & 0 & 0 & 0 & 0 & 1 & - \frac{1}
     {3}   & 0 & - \frac{1}{3}   & 0 & 0 & 0 & 0 \cr 0 & 0 & 0 & 0 & 0 & 0 & 0 & 0 & - \frac{1}{3}
       & 0 & 1 & 0 & 0 & 0 & - \frac{1}{3}   & 0 & - \frac{1}{3} 
      \cr 0 & 0 & 0 & 0 & 0 & 0 & 0 & 0 & 0 & 0 & 0 & 1 & 0 & 0 & 0 & - \frac{1}{2}   & - \frac{1}{2} 
      \cr 0 & 0 & 0 & 0 & 0 & 0 & 0 & 0 & 0 & 0 & 0 & 0 & 1 & 0 & 0 & - \frac{1}{2}   & - \frac{1}{2} 
      \cr 0 & 0 & 0 & 0 & 0 & - \frac{1}{3}   & 0 & - \frac{1}{3} 
      & 0 & 0 & 0 & 0 & 0 & 1 & 0 & 0 & - \frac{1}{3}   \cr 0 & 0 & 0 & 0 & 0 & 0 & 0 & 0 & - \frac{1}{3}
       & 0 & - \frac{1}{3}   & 0 & 0 & 0 & 1 & 0 & - \frac{1}{3} 
      \cr 0 & 0 & 0 & 0 & 0 & 0 & 0 & 0 & 0 & 0 & 0 & - \frac{1}{3}   & - \frac{1}{3} 
      & 0 & 0 & 1 & - \frac{1}{3}   \cr 0 & 0 & 0 & 0 & 0 & 0 & 0 & 0 & 0 & 0 & 0 & 0 & 0 & 0 & 0 & 0 & 1
      \end{array}\right) h_{.f} (\mathcal G)= \left(\begin{array}{c}1\\1\\1\\1\\1\\1\\1\\1\\1\\1\\1\\1\\1\\1\\1\\1\\0\end{array}\right)$$
}
When solving this system, we obtain that 

$h_{.f}(\mathcal G)=\left(\frac{34}{5};\frac{109}{20};\frac{13}{2};\frac{109}{20};\frac{49}{10};4;\frac{13}{2};5;5\frac{49}{10};4;\frac{9}{4};\frac{9}{4};4;4;\frac{5}{2};0\right)$

Thus, $C_1=\frac{34}5$

\subsection{Efficient cover time computation}

The matrix to be inverted in the previous method is large (about $n2^n\times n2^n$: we only provide an upper bound since the graph size can be reduced by suppressing the unreachable states), leading to a complexity approaching $n^38^n$.

However, this graph has some particularities that we want to exploit in order to improve the efficiency of the computation. The subgraphs constituted by all the vertices $(P, i)$ with the same $P$ are undirected. The time it takes to reach $f$ from $i$ when $P$ is the set of already visited vertices in $\mathcal G$ can be decomposed from  the time it takes  to reach the first vertex $j$ out of $P$ plus the expected time from $(j, P\cup\{j\})$ to $f$ (the expectation being computed over all possible $j$ wrt their probabilities of being the first hitten vertex outside of $P$). Thus, the cover time can be computed according to:

$$h_{(P, i), f}=1+s(P, i)+\sum_j p(P,i, j) h_{(P\cup\{j\}, j), f}$$

where
\begin{itemize}
\item $p(P,i, j)$ is the probability that the first vertex outside $P$ hitten by a random walk starting at $i$ is $j$;
\item $s(P, i)$ the average time the walk starting at $i$ stays in $P$
\end{itemize}

$1+s(P, i)$ is the expected time the walk will spend in the strongly connected component defined by $P$, when it is on $i$.
The next newly visited site is $j$ with probability $p(P,i, j)$, and once on this site, the walk will take an expected time of $h_{(P\cup\{j\}, j), f}$ to achieve the coverture. 

Thus, the equation above can be decomposed in $1+s(P, i)$ which represents the time spent in a strongly connected component and $\sum_j p(P,i, j) h_{(P\cup\{j\}, j), f}$ which represents the expected time to reach $f$ in the directed acyclic graph of strongly connected components.

We can express both of those quantities in terms of equivalent resistances and potentials, making it possible to use results from the previous section: 
for any $i$ in $P$ ($i$ represents the current location of the token) and $j$ in $V\backslash P$ ($j$ represents the first site the token will reach outside $P$)

\begin{itemize}
\item $p(P,i, j)$ is the potential in $i$ when $V_j=1$ and all other sites in $V\backslash P$ have potential 0
\item $s(P, i)=h_{i(V\backslash P)}(G(P\cup\{j\}))$
\end{itemize}
Those quantities can be computed thanks to a $|P|\times |P|$ matrix inversion.

Indeed, we have already remarked that the potential on one node $c$, when a given node $a$ has potential 0 and another $b$ has potential 1 is the probability to hit $b$ before $a$ when the current node is $c$. Thus, the potential on $i$ when $V_j=1$ and all other sites in $V\backslash P$ have potential 0 is the probability that the next newly visited vertex is $j$.

$s(P, i)$ is the average time the token spends in $P$, since it is the expected time to reach a node in $V\backslash P$.

In fig.\ref{gr5},  we represented $\cal G$ and circled 
subgraphs that are not directed. Each of them is also a subgraph 
(connected and containing 1) of $G$. We have to compute the time the 
random walk spends in each of the subgraphs, considering its arrival 
point. In fig \ref{gr6}, we highlighted the directed edges joining the various 
subgraphs: each of them represents the discovery of a new vertex. We 
have to compute the probability that the walk crosses each of these 
edges, depending on the vertex of the subgraph it arrives on. Then, 
using those information, we can compute the cover time with the above 
formula.

\begin{figure}
\begin{center}
\begin{minipage}[b]{.45\linewidth}
\includegraphics[angle=270, width=.9\linewidth]{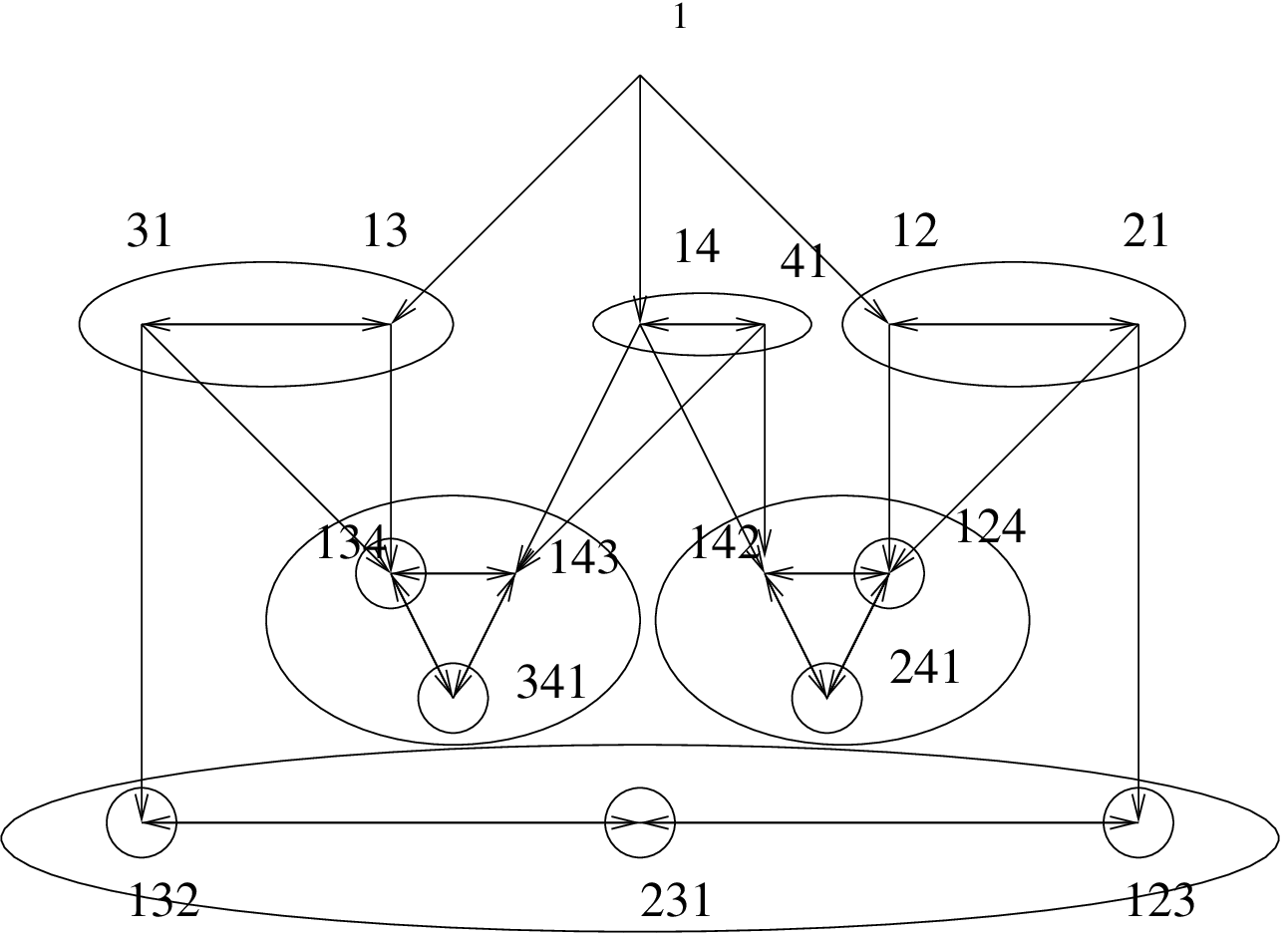}
\caption{$G$\label{gr5}}
\end{minipage}
\begin{minipage}[b]{.45\linewidth}
\includegraphics[angle=270, width=.9\linewidth]{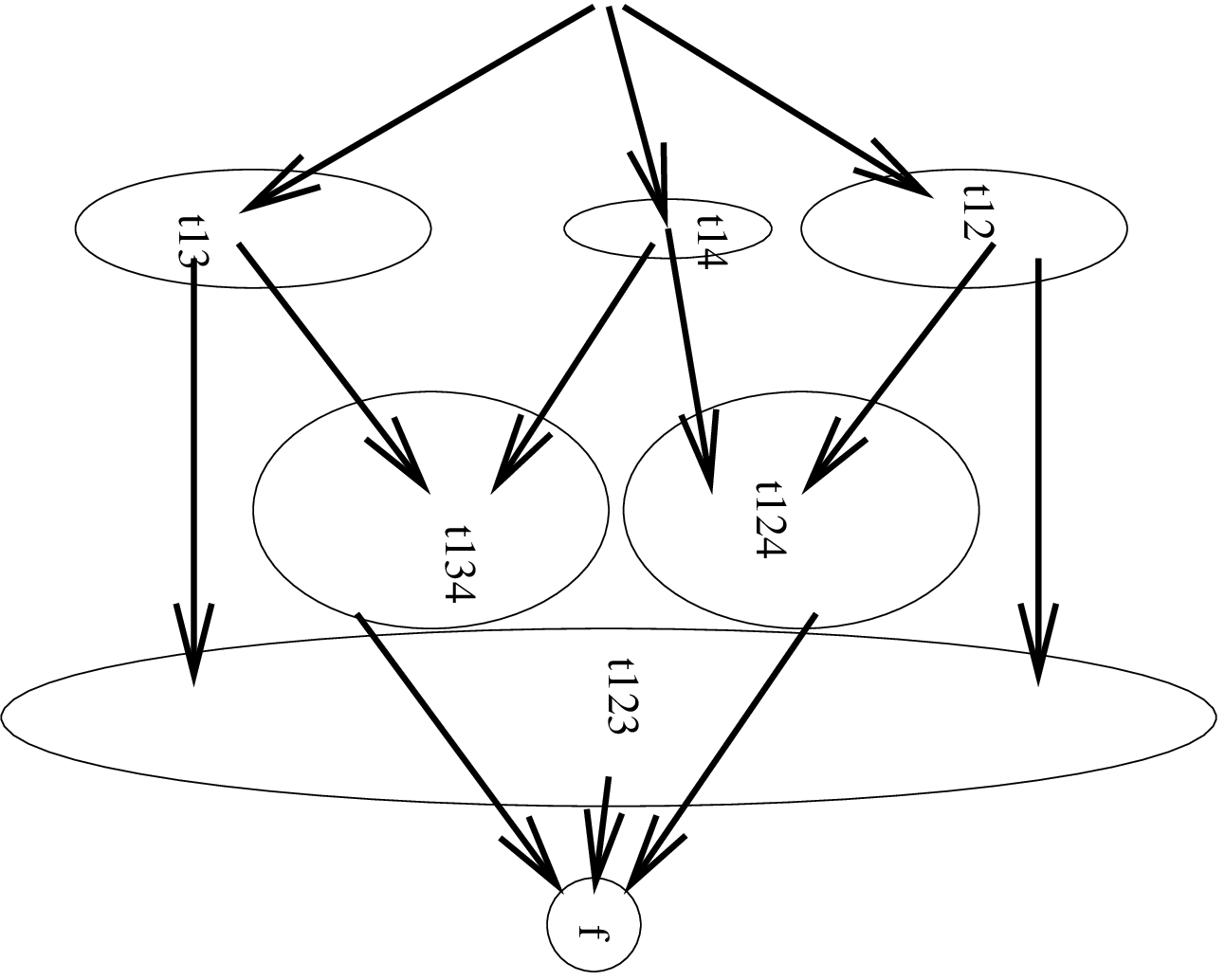}
\caption{$\mathcal G$\label{gr6}}
\end{minipage}
\end{center}
\end{figure}

The complexity of this procedure is one $k\times k$ matrix inversion for each subgraph of size $k$ appearing. The complexity is then at most $O(n^32^n)$. However, it highly depends on the topology of $G$. If $G$ is a chain, only $n$ subgraphs appear (a subgraph occurs in the computation \emph{iff} it contains the state $1$ and is connected), and the complexity is $O(n^4)$.

This rather expensive computation is robust to the topological evolution of the graph, thanks to the Rayleigh's shortcut principle exposed above. The cover time computation being based on a hitting time computation, adding or removing a limited number of edges in the graph do not modify the cover time by more than the ratio of the weights of modified edges to the global weight of the graph.

\section{Conclusion}

Random walk based algorithms represent an important class of distributed algorithms, two of their main features  are that they require no assumptions on the topology of the network and that they can easily handle topological changes without any special procedure triggered by a change. 
The exact computation of hitting and cover times allows the computation of the complexity of these algorithms.

Further research can be conducted based on the exact computation of the hitting and cover times. These results are more clearer than previous results which were approximation.  
We plan to overview the hitting and cover times over various topologies, ranging from classical topologies, like hypercubes or tori, to topologies modeling the actual high-scale distributed systems, like small-world graphs, some categories of random graphs and maps of parts of peer-to-peer file-sharing networks. We hope this work will provide tracks on the topologies to consider in order to achieve a good behavior of the walk and on the impact of a slight difference between the actual topology of a network and the intended topology.

Thus, the hitting and cover times allow  us determine the complexity of a wide class of algorithms, but we can also improve them by choosing the topologies in which they are efficient.

\bibliographystyle{fundam}
\bibliography{./biblio}

\appendix

\section*{Appendices}

\subsection*{The Cover Time of the Complete Graph}

Let $G=(V,E)$ be the complete graph on $n$ vertices. Let $C$ be the average time a random walk on $G$ takes to visit every vertex in $G$ (the \emph{cover time}; note that the starting vertex does not matter here, since the graph is symmetric). Let $C_k$ be the average time the random walk takes to visit the $k+1$-st vertex when it has visited the $k$-th one.

Then,\[C=\sum_{k=1}^{n-1}C_k\]

When the walk has visited $k$ vertices, at the next step, it has $\frac{n-k}{n-1}$chance to visit a new vertex, and $\frac{k-1}{n-1}$ chance to visit an already known one. Thus, the expected time to visit a new vertex, when $k$ vertices have already been visited is:
\begin{align*}
C_k &= \sum_{i\in\mathbb N}i\left(\frac{k-1}{n-1}\right)^{i-1}\frac{n-k}{n-1}
= \sum_{i\in \mathbb N}\sum_{j=1}^{i}\left(\frac{k-1}{n-1}\right)^{i-1} \frac{n-k}{n-1}\\
&= \sum_{j\in\mathbb N^*}\left(\sum_{i\geq j}\left(\frac{k-1}{n-1}\right)^{i-1}\right)\frac{n-k}{n-1}
= \sum_{j\in\mathbb N^*}\left(\frac{k-1}{n-1}\right)^{j-1}\frac1{1-\frac{k-1}{n-1}}\frac{n-k}{n-1}\\
&= \left(\frac1{1-\frac{k-1}{n-1}}\right)^2\frac{n-k}{n-1}
= \left(\frac{n-1}{n-k}\right)^2\frac{n-k}{n-1}\\
&= \frac{n-1}{n-k}
\end{align*}

Then:
\begin{align*}
C&=\sum_{k=1}^{n-1}C_k
=\sum_{k=1}^{n-1}\frac{n-1}{n-k}
=(n-1)\sum_{k=1}^{n-1}\frac1{n-k}
=(n-1)\sum_{i=1}^{n-1}\frac1i\\
&=(n-1)H_{n-1}
\end{align*}
with $H_n$ the $n$-th harmonic number.

Thus $$C\sim_{n\rightarrow \infty} n\log n$$

\end{document}